\documentclass[letterpaper,10pt,conference]{ieeeconf}

\usepackage{mathtools}
\usepackage{amssymb, amsmath, bm}
\usepackage{xcolor}
\usepackage{balance} 
\usepackage{algorithm}
\usepackage{float}
\usepackage[noend]{algpseudocode}
\usepackage{amsmath, amsfonts}
\usepackage{adjustbox}
\usepackage{threeparttable}

\usepackage{amsthm}

\usepackage{booktabs}

\usepackage{hyperref}
\usepackage{cleveref}
\let\labelindent\relax
\usepackage{enumitem}

\DeclareMathOperator*{\argmax}{arg\,max}

\usepackage[style=ieee,dashed=false]{biblatex}
\IEEEoverridecommandlockouts                             
\theoremstyle{plain}
\newtheorem{theorem}{Theorem}
\newtheorem{lemma}{Lemma}

\newtheorem{proposition}{Proposition}

\theoremstyle{definition}
\newtheorem{definition}{Definition}

\theoremstyle{remark}

\crefname{theorem}{Theorem}{Theorems}
\crefname{algorithm}{Algorithm}{Algorithms}
\crefname{section}{Section}{Sections}
\crefname{lemma}{Lemma}{Lemmas}
\crefname{proposition}{Proposition}{Propositions}
\crefname{definition}{Definition}{Definitions}

\allowdisplaybreaks

\title{ \bf
{Dynamic Auction Mechanisms in a Closed-Credit Economy: \\  Sequential Slot Reallocation in Air Traffic Management}}

\author{Jihun Kim, Pan-Yang Su, Maria G. Mendoza, Victoria Marie Tuck, Chinmay Maheshwari, S. Shankar Sastry
\thanks{This work was supported by the National Aeronautics and Space Administration (NASA). Jihun Kim, Pan-Yang Su, Maria G. Mendoza, and S. Shankar Sastry are with the University of California, Berkeley.  Victoria Marie Tuck is with the University of Pennsylvania. Chinmay Maheshwari is with the Johns Hopkins University. Correspondence to: {\tt\footnotesize jihun.kim@berkeley.edu}}
}

\begin{document}

\maketitle 
\thispagestyle{empty}
\pagestyle{empty}

\begin{abstract}
Indivisible resources allocated in advance may need to be reallocated when capacity reductions make existing assignments infeasible, yet monetary compensation is often impractical or undesirable. We develop a sequential resource reallocation scheme using reusable artificial credits, where agents hold pre-existing resource rights and have heterogeneous marginal valuations over time blocks. Capacity reductions produce fractional endowments, which our closed-credit scheme resolves into indivisible final allocations. We model agents' intertemporal preferences through a predicted-value-to-go function and characterize conditions ensuring that their optimal demand is non-increasing in price. We propose an ascending-auction mechanism that partitions agents into subgroups and uses the minimum subgroup market-clearing price as a common transaction price. We show that the mechanism preserves non-myopic individual rationality and weak nontrader neutrality. Moreover, we establish Pareto efficiency under optimal-demand clearing in the unified market and bound the expected gain from misreporting by the probability that an agent changes the common price. Simulations of airport slot reallocation demonstrate reductions in total modeled delay and cancellation costs relative to a ground-delay baseline.
\end{abstract}

\section{Introduction}

We study dynamic auction markets in which long-lived agents enter with existing resource endowments, resource capacities are uncertain, and monetary transfers are restricted. Such problems arise when scarce resource rights are allocated in advance but must subsequently be adjusted as uncertain capacities are realized. Unlike a sequence of
static allocation problems, previous allocations create endowments:
agents retain rights to previously allocated resources, so subsequent
reallocation requires exchange among agents. This intertemporal exchange
structure arises naturally in transportation capacity, 
advance reservations, and other settings in which usage rights are
allocated before uncertainty is fully resolved.

Our goal is to design a mechanism with desirable economic
properties, including Pareto efficiency, incentive
compatibility, individual rationality, and weak nontrader
neutrality, a counterpart of weak budget balance in our
model. Our mechanism builds on artificial currencies.
Although credits have no intrinsic value, they can acquire
endogenous value by providing access to resources in
future markets. We capture this future benefit through
a predicted-value-to-go function that enters agents'
utilities and influences their willingness to trade.
Existing artificial-currency mechanisms study repeated
access to shared resources \cite{elokda2024self,288635}.
Our setting focuses on reallocating block-specific
resource rights assigned before capacity uncertainty
is resolved. Agents trade relative to their scaled
endowments, balancing current resource values against
the continuation value of credits. This creates a
dynamic exchange problem with capacity uncertainty
and pre-existing resource rights.

\textbf{Theoretical difficulties.}
Even when predicted-value-to-go functions are affine
with strictly positive slopes, which yield quasilinear preferences over current allocations
and net credit transfers, the mechanism-design problem
remains challenging. Unlike bilateral trade and conventional double auctions \cite{myerson1983efficient,MCAFEE1992434,SegalHalevi2018MUDA}, agents are not exogenously designated as buyers or sellers. Each agent enters with a resource endowment and may become a buyer or a seller depending on its valuation, credit balance, and the market price. Moreover, unlike standard multi-unit auctions with budget constraints \cite{DOBZINSKI2012486}, the mechanism must accommodate both sides of the market while respecting agents' credit constraints.

Our first result, \cref{thm:impossibility}, demonstrates a fundamental limitation of this exchange problem. We show that Pareto efficiency, incentive compatibility, invidual rationality, and weak nontrader neutrality cannot in general be achieved simultaneously. The result follows by embedding bilateral trade as a special case and applying the Myerson--Satterthwaite impossibility theorem \cite{myerson1983efficient}.

\textbf{Proposed mechanism.}
Motivated by this impossibility, we propose an ascending-auction mechanism (\cref{alg:main_auction}) that preserves individual rationality and weak nontrader neutrality. We characterize conditions on the
predicted-value-to-go functions that ensure price-monotone
demand (\cref{prop:optimal_demand_price_monotonicity}), which lays the foundation for ascending price discovery.
Building on budget-constrained multi-unit auctions
\cite{DOBZINSKI2012486} and multi-unit double auctions
\cite{SegalHalevi2018MUDA}, the mechanism partitions agents
into $G$ subgroups, conducts price discovery within each subgroup,
and uses the minimum subgroup price as the common transaction
price. Subgrouping restricts trading opportunities but can
also limit an agent's influence on the common price.
We establish a conditional Pareto-efficiency guarantee when $G=1$ (\cref{thm:mechanism_properties2}) and bound manipulation gains in terms
of the probability of changing the common price (\cref{thm:approximate-NIC}). All proofs are deferred to the supplementary material.

\textbf{Airport slot allocation.}
We apply our mechanism to airport slot allocation
\cite{vossen_ball_2006}.
Airlines plan flights well in advance, while available
capacity remains uncertain. Using our mechanism, we develop a day-of market that uses reusable credits to elicit airlines' demands and compensate them for releasing capacity across successive time blocks. We evaluate its performance through simulations.




\section{Related Work}
\textbf{Dynamic allocation without money.}
Resource allocation with artificial currencies has a long
history, including pseudomarkets \cite{doi:10.1086/260757}
and approximate competitive equilibrium from equal incomes
\cite{doi:10.1086/664613}. A growing literature extends
this idea to repeated resource-allocation problems,
including spot mechanisms
\cite{https://doi.org/10.3982/TE5419} and dynamic approximate
competitive equilibrium from equal incomes
\cite{nguyen2023dynamiccombinatorialassignment}.
These approaches typically focus on repeated resource allocation problems, while we deal with exchange problems where agents have endowments. Banerjee et al.~\cite{10.1145/3580507.3597723} study
reusable resources with unit demand within each round
and potentially multi-round usage requirements.
Our agents may demand multiple units within a block,
with different marginal valuations, and trade relative
to pre-existing resource endowments.

Karma mechanisms \cite{elokda2024self,288635} are
particularly related to our setting, where they use artificial currencies, called ``karma,'' to coordinate resource allocation. Their models, however,
differ from ours. Elokda et al.~\cite{elokda2024self} study repeated
competition for a resource with heterogeneous urgency
but without initial resource endowments, while
Vuppalapati et al.~\cite{288635} incorporate fair shares
but focus on demand quantities rather than heterogeneous
marginal valuations. Our model combines heterogeneous
multi-unit valuations with credit-based exchange of
resource rights under capacity uncertainty.

A related literature studies dynamic exchange, including
\cite{10.1111/j.1467-937X.2009.00575.x,2887007.2887135}. These models,
however, are motivated primarily by organ exchange and have a different
allocation structure from ours. In particular, organ-exchange models
typically focus on one-to-one or compatibility-constrained matching, while our application involves allocating capacity to agents with
potentially multi-unit demand.

Methodologically, our approach is closely related to
\cite{doi:10.1287/moor.2020.1098}, which uses artificial currencies to
translate monetary mechanisms into approximately incentive-compatible
nonmonetary mechanisms in repeated allocation environments. Our setting
differs because previously allocated rights become endowments and
reallocation therefore takes the form of exchange. As a result, the
guarantees developed for repeated assignment problems do not carry over directly.

\textbf{Bilateral Exchange and Double Auctions.}
The difficulty of exchange under private information was first formalized by Myerson and Satterthwaite
\cite{myerson1983efficient}, in what is now known as the
Myerson--Satterthwaite theorem. Subsequent work extends this
impossibility to richer exchange environments, including settings with multiple goods or multiple agents
\cite{https://doi.org/10.1111/jpet.70109,10.1609/aaai.v38i9.28857}.
More recently, \cite{11312856} studies a two-stage market in which
previously allocated resources are subsequently exchanged and derives
related impossibility results. Our work is closely related to this line
of research, but considers exchange mediated by an artificial currency
whose value derives from future allocation opportunities.

People have proposed methods to get around the impossibility results by trading off different properties. Our work is related to double auctions \cite{MCAFEE1992434}. Recently, Segal-Halevi et al. proposed a truthful multi-unit double-auction mechanism \cite{SegalHalevi2018MUDA}. They split the agents into two groups, run auctions separately, and exchange the prices. This mechanism ensures truthfulness. Our model is more difficult in two aspects. First, their models assume that the roles are exogeneously given, but our agents can be either sellers or buyers depending on the prices. Second, our agents have budget constraints, so naively swapping the prices may result in budget constraint violation.



\textbf{Airspace market design.}
Previous work considers mediated slot trading
\cite{vossen_ball_2006},
preference-based slot scheduling
\cite{doi:10.1287/trsc.2019.0926},
combinatorial exchanges \cite{PELLEGRINI20121009},
and auction-based slot reallocation \cite{lee2024auction}.
These approaches either rely on real-money transfers, which may be undesirable or impractical in operational airspace markets, or lack an alternative currency to facilitate exchange.
Our work bridges this gap by adopting artificial credits as a medium of exchange, enabling flexible slot reallocation without real-money transactions.

Credit-based approaches have also been explored in airspace management, albeit for different purposes. In
\cite{sheth2008incorporating,doi:10.2514/6.2009-7128},
airlines assign credits to flights to influence rerouting
and delay decisions, but not theoretical analysis is given. In \cite{doi:10.2514/6.2010-8078},
these credits serve as flight-specific delay-cost weights
in a traffic flow optimization model, which is different from our use an artificial currency.

More recently, Im et al. \cite{IM2027106051} study
decentralized negotiation through asset trading,
using regulatory oversight to guide consensus toward
efficient and equitable outcomes. Their focus is
coordination over a common decision, whereas ours
is the sequential reallocation of scarce capacity
among agents with pre-existing rights. Thus, applying their negotiation mechanism to our setting is computationally infeasible.

\section{Model}
\label{sec: overview}
We study sequential reallocation of indivisible resources
when capacity shocks may leave too few items to fulfill
agents' initial allocations (\emph{endowments}).
We scale these endowments proportionally to the remaining
capacity and design a mechanism that allocates integer
quantities while accounting for the scaled entitlements
through credit transfers. Agents receiving fewer items
than their scaled entitlements receive credits, while
those receiving more pay credits. Credit balances carry
forward for use in subsequent blocks.

\subsection{Preliminaries}
Consider a set $I$ of agents and a set $B$ of time blocks. Each block $b\in B$ consists of a collection of discrete items, and we refer to the number of items as the \emph{capacity} of the block. Block $b$ has \emph{nominal capacity} $C_b$ and realized \emph{operational capacity} $\bar C_b\leq C_b$.

Agent $i\in I$ initially holds its initial allocations for each block $b\in B$, represented by $e_i^b\in\mathbb Z_+$ items in block $b$, which we refer to as its \emph{initial endowment}. The total initial endowments do not exceed the nominal capacity:
$\sum_{i\in I}e_i^b \leq C_b.$

Following a capacity reduction, we scale each agent's endowment $e_i^b$ proportionally and define its \emph{scaled endowment} as $\bar e_i^b= e_i^b \cdot  \bar C_b / C_b$, which clearly satisfies $\sum_{i\in I} \bar e_i^b\le \bar C_b$. 
While $e_i^b$ is integer-valued, $\bar e_i^b$ need not be. We interpret $\bar e_i^b$ as an accounting entitlement denominated in items rather than a physical allocation; the final allocation of items remains integer-valued. 

Each agent $i$ enters block $b$ with a \emph{credit balance} $w_i^b\in\mathbb R_+$. The initial balance $w_i^1$ is given, while the balances in subsequent blocks are determined endogenously by the outcomes of earlier blocks. Fig.~\ref{fig:airport_slot_reallocation} illustrates our sequential resource reallocation scheme under such a closed-credit economy. 

\begin{figure}
    \centering
    \includegraphics[width=1\linewidth]{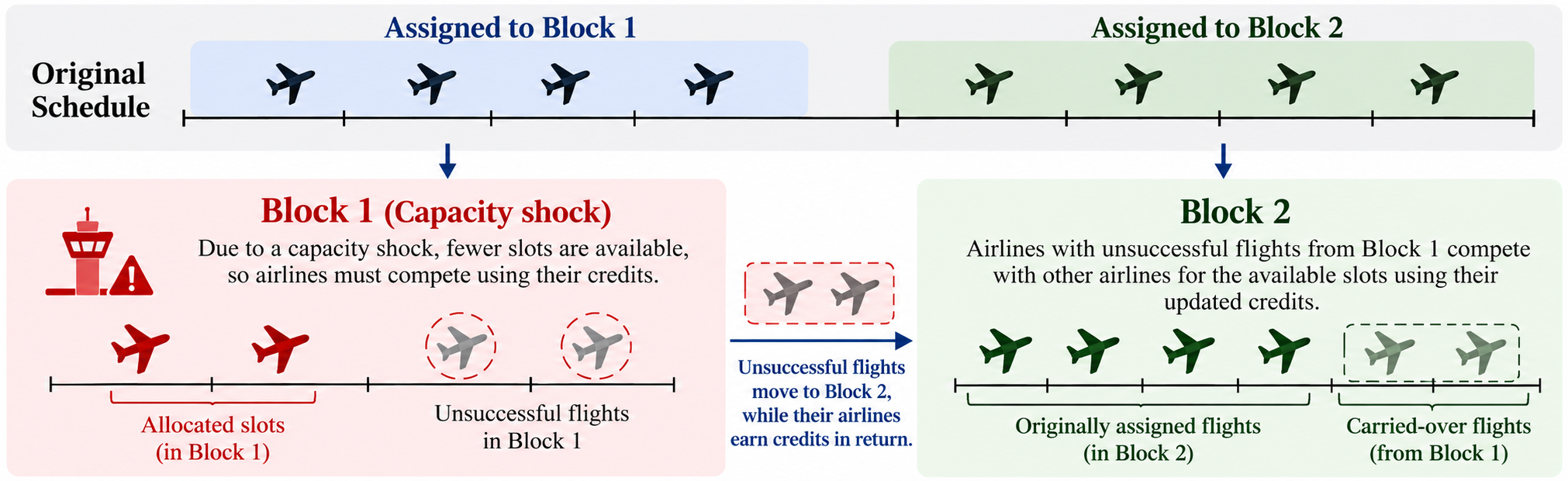}
    \caption{A capacity reduction limits available indivisible items, prompting agents to compete using credits. Unallocated demand rolls over to the next block, while unused credits are carried forward. Here, agents represent airlines and items represent flight slots.}
    \label{fig:airport_slot_reallocation}
\end{figure}

For each block $b$, the market determines a \emph{final allocation} $\{x_i^{b,*}\}_{i\in I}$ and a \emph{clearing price} $p_b^*\geq0$. Let $N_i^b\in\mathbb Z_+$ denote the maximum number of items that agent $i$ can operate in block $b$. Then, the allocation should satisfy
$x_i^{b,*}\in\{0,1,\ldots,N_i^b\},$
and all units in block $b$ are traded at the same unit price $p_b^*$. Relative to its scaled endowment, agent $i$ makes the net credit payment\footnote{The payment \(T_i^b\) in \eqref{eq:net_payment} can take both positive or negative values, denoting payment or rebate, respectively.}
\begin{equation}\label{eq:net_payment}
    T_i^b
    =
    \bigl(x_i^{b,*}-\bar e_i^b\bigr)p_b^*.
\end{equation}
Thus, an agent that acquires more capacity than its scaled endowment pays credits, whereas an agent that releases endowments receives credits. Its credit balance therefore evolves according to
\begin{equation}\label{eq:budget_transition}
    w_i^{b+1}
    =
    w_i^b-T_i^b
    =
    w_i^b-\bigl(x_i^{b,*}-\bar e_i^b\bigr)p_b^*.
\end{equation}
This transition connects decisions across blocks:
spending credits in the current block reduces an agent's purchasing power in future blocks, while releasing endowments increases it.

For each block $b$, the outcome must respect the available
capacity and agents' credit budgets. The budget constraints
require net purchases to be affordable and are automatically
satisfied when $x_i^b\leq\bar e_i^b$.

\begin{definition}[Feasible Allocation-Price Pair]\label{def:feas_alloc}
Given realized capacity $\bar C_b$, credit balances $\{w_i^b\}_{i\in I}$, and scaled endowments $\{\bar e_i^b\}_{i\in I}$, an allocation-price pair $(\{x_i^b\}_{i\in I},p_b)$ is \emph{feasible} if the following constraints hold:
\medskip
\begin{enumerate}
    \item Price Nonnegativity: $p_b \geq 0$.
    \item Domain constraints: $x_i^b \in\{0,1,\dots,N_i^b\}$ for all $i \in I$.
    \item Budget constraints: $(x_i^b-\bar e_i^b)p_b \leq w_i^b$ for all $i \in I$.
    \item Capacity constraints: $\sum_{i\in I}x_i^b \leq \bar C_b$.
\end{enumerate}
\end{definition}

Recall that $\bar e_i^b$ need not be integer-valued. It represents an accounting entitlement rather than a physical allocation: the final allocation $x_i^b$ remains integer-valued, while $\bar e_i^b$ determines the agent's net trading position and hence its net credit payment. Thus, fractional physical items are never allocated, and any non-integer baseline is resolved solely through financial credit adjustments.

\subsection{Predicted Non-myopic Value}\label{sec:value}
Let $v_i^b(k)\in\mathbb R_+$ denote agent $i$'s \emph{marginal value} for its $k$th item in block $b$. Let $J_i^{b+1}(w)$ denote agent $i$'s \emph{predicted-value-to-go}, evaluated at block $b$, from entering block $b+1$ with credit balance $w$, which is nondecreasing in $w$.\footnote{In the spirit of approximate dynamic
programming, $J_i^{b+1}$ approximates the agent's
expected future utility as a function of its credit
balance, with other state variables held fixed.
It need not coincide with the true value-to-go.} Thus, any future utility is absorbed into $J_i^{b+1}$. 
Credits have no intrinsic value. Their value arises endogenously from the future allocation opportunities they provide.

Consider block $b$ at a unit price $p\geq0$. Suppose agent $i$ enters the block with credit balance $w$ and scaled endowment $e$. If it chooses a final allocation of $K$ items, its credit balance in the next block is $w-(K-e)p$. The agent therefore solves the following \textit{predicted non-myopic value} maximization problem:
\begin{equation}\label{eq:dynamic_demand}
\begin{aligned}
    \max_{K\in\{0,1,\dots,N_i^b\}}
    \quad&
    \sum_{k=1}^K v_i^b(k)
    +
    J_i^{b+1}\!\left(w-(K-e)p\right)\\
    \text{s.t.}\quad&
    (K-e)p\leq w.
\end{aligned}
\end{equation}
The constraint requires any net purchase to be affordable. When $K<e$, the agent instead releases capacity and receives $(e-K)p$ credits.

The marginal gain in the objective from increasing the allocation from $k-1$ to $k$ is
\begin{equation}\label{eq:effective_marginal}
\begin{aligned}
    \Delta_i^b(k,p,w,e)
    \equiv\;&
    v_i^b(k)
    -[J_i^{b+1}\!\left(w-(k-1-e)p\right)\\
    &\hspace{25mm}-J_i^{b+1}\!\left(w-(k-e)p\right)].
\end{aligned}
\end{equation}
The bracketed difference in predicted-value-to-go $J_i^{b+1}(\cdot)$ captures the opportunity cost of spending additional credits in the current block. Thus, subject to budget feasibility $(K-e)p \le w$, the $k$th item is weakly preferred to stopping at $k-1$ if and only if
\begin{equation}\label{eq:effective_positive}
    \Delta_i^b(k,p,w,e)\geq0.
\end{equation}
Under the following conditions, \eqref{eq:effective_positive}
characterizes an optimal demand.

\begin{proposition}[Optimal Demand]\label{prop:optimal_demand}
Suppose the marginal valuation $v_i^b(k)$ is nonincreasing in $k$ and the predicted-value-to-go $J_i^{b+1}(w)$ is concave in the credit balance $w$. Then, given the price $p$, budget $w$ and endowment $e$, the marginal gain $\Delta_i^b(k,p,w,e)$ is nonincreasing in $k$, and the largest feasible $k\geq 1$ satisfying \eqref{eq:effective_positive} is an optimal demand (with optimal demand $0$ if no such $k$ exists).
\end{proposition}

Economically, the conditions of Proposition \ref{prop:optimal_demand} ensure that as the agent acquires more items, the marginal utility of the next item falls while the opportunity cost of the required credits rises.

\subsection{Price Monotonicity of Optimal Demand}
The predicted-value-to-go function $J_i^{b+1}$ determines the endogenous value of credits for future time blocks. In this section, we characterize the class of predicted-value-to-go functions under which optimal demand  is nonincreasing in the price $p$.
This price monotonicity serves as a core foundation for our mechanism design.

\begin{proposition}[Price Monotonicity of Marginal Gains]
\label{prop:nonincreasing_in_p}
Let $J:[0,\infty)\to\mathbb R$.
For $w\geq0$ and $a\in\mathbb R$, define
\[
F_{w,a}(p)
=
J(w-ap)-J(w-(a+1)p)
\]
on the feasible price interval
\[
\mathcal P_{w,a}
=
\{p\geq0:w-(a+1)p\geq0\}.
\]
Then $F_{w,a}$ is nondecreasing on $\mathcal P_{w,a}$
for every $w\geq0$ and $a\in\mathbb R$ if and only if
$J$ is concave on $[0,\infty)$ and $z\longmapsto zJ'_+(z)$
is nondecreasing on $(0,\infty)$, where $J'_+$
denotes the right derivative of a concave function.
\end{proposition}

Taking $J=J_i^{b+1}$ and $a=k-1-e$, we have
\[
\Delta_i^b(k,p,w,e)
=
v_i^b(k)-F_{w,k-1-e}(p).
\]
Therefore, under the conditions of
\cref{prop:nonincreasing_in_p},
the marginal gain is nonincreasing in the price
on its feasible domain.

The condition of concavity corresponds to a diminishing marginal value of additional wealth. The additional condition that $w \mapsto w J_+'(w)$ is nondecreasing implies that this marginal value cannot decay too quickly as wealth increases. Standard examples satisfying these conditions include affine functions, power functions $w^\alpha$ for $\alpha \in (0,1]$, and logarithmic functions $\ln(w + b)$ for $b > 0$. We define such a function class below.

\begin{definition}
Let $\mathcal{J}$ denote the class of concave functions
$J:[0,\infty)\to\mathbb{R}$ for which the map
$z\mapsto zJ'_+(z)$ is nondecreasing on $(0,\infty)$.
\end{definition}

Combining Propositions \ref{prop:optimal_demand} and \ref{prop:nonincreasing_in_p} yields the price monotonicity of optimal demand.

\begin{proposition}[Price Monotonicity of Optimal Demand]\label{prop:optimal_demand_price_monotonicity}
Suppose that the marginal valuation $v_i^b(k)$ is nonincreasing in $k$ and $J_i^{b+1}(w) \in \mathcal{J}$. Then, for a given budget $w$ and endowment $e$, the largest optimal demand (defined in Proposition \ref{prop:optimal_demand}) is nonincreasing in the price $p$.
\end{proposition}

Within the class $\mathcal{J}$, consider the special case of affine predicted-value-to-go functions $J_i^{b+1}(w) = \alpha_i^{b+1} w + \beta_i^{b+1}$. Here, the threshold rule \eqref{eq:effective_positive} simplifies to:
\begin{equation}\label{eq:quasilinear_marginal}
  \Delta_i^b(k,p,w,e)=
    v_i^b(k)-\alpha_i^{b+1} p\ge 0
    \quad\Longleftrightarrow\quad
    \frac{v_i^b(k)}{\alpha_i^{b+1}}\ge p. 
\end{equation}

Thus, under an affine predicted-value-to-go function, the agent follows the standard quasilinear demand rule with credit-denominated marginal values $v_i^b(k)/\alpha_i^{b+1}$ independent of $(w,e)$,  subject to budget feasibility. The coefficient $\alpha_i^{b+1}$ of the predicted-value-to-go function acts as the endogenous conversion rate between item value and future credits.

\section{Problem Formulation}
Under the model in \cref{sec: overview}, a mechanism is a mapping from the primitives to a feasible allocation-price pair. Our goal is to design a mechanism such that some desirable (\emph{normative}) properties are satisfied. 

\subsection{Normative Properties}
\label{sec: normative}
\textbf{Non-myopic individual rationality.}
We first define non-myopic individual rationality relative to the administrative allocation that would arise in the absence of trade. Since the scaled endowment $\bar e_i^b$ need not be integer-valued, write
\[
    \bar e_i^b=\lfloor \bar e_i^b\rfloor+r_i^b,
    \qquad r_i^b\in[0,1).
\]
If $r_i^b=0$, agent $i$ retains $\bar e_i^b$ items without trading. If $r_i^b\in(0,1)$, we take the administrative baseline to be the allocation obtained by rounding the scaled entitlement down to $\lfloor \bar e_i^b\rfloor$ items. Accordingly, participation in each time block is evaluated relative to this floor-based administrative allocation, with the fractional residual $r_i^b$ settled at the clearing price $p_b^*$.

\begin{definition}[Non-myopic Individual Rationality]
A mechanism for block $b$ satisfies \textbf{Non-myopic Individual Rationality (NIR)} if, for every resulting allocation $x_i^b$ with a clearing price $p_b^*$ and net credit payment $T_i^b$, every agent $i\in I$ obtains a non-myopic utility inclusive of its predicted-value-to-go $J_i^{b+1}$ weakly greater than its administrative baseline\footnote{The rebate $r_i^b p_b^*$ may be added to the budget in the administrative baseline as is done later. The comparison statement holds in both cases.}:
\begin{equation}\label{eq:block_ir}
    \sum_{k=1}^{x_i^b}v_i^b(k)
    +
    J_i^{b+1}(w_i^b-T_i^b)
    \geq
    \sum_{k=1}^{\lfloor\bar e_i^b\rfloor}v_i^b(k)
+
J_i^{b+1}(w_i^b).
\end{equation}
\end{definition}

Non-myopic IR guarantees that participating in each block market does not make an agent worse off than the corresponding non-trading administrative baseline.

\textbf{Weak nontrader neutrality.}
In a standard open-money economy, Weak Budget Balance (WBB) ensures that the auctioneer does not face a deficit. We retain the analogous accounting condition for credit transfers.

\begin{definition}[Weak Budget Balance]
A mechanism satisfies Weak Budget Balance (WBB) if, for every block $b$, the aggregate net credit payment collected from the agents is nonnegative:
\begin{equation}\label{eq:wbb}
    \sum_{i\in I}T_i^b
    =
    \sum_{i\in I}(x_i^b-\bar e_i^b)p_b^*
    \geq0.
\end{equation}
\end{definition}

In a closed-credit economy, however, the relevant concern is not an external monetary deficit but the effect of credit creation on agents' future purchasing power. If the mechanism injects credits to finance a deficit in the current block, total outstanding credits increase while the balances of nontrading agents remain unchanged. Their relative purchasing power is therefore diluted. Because credits carry over across blocks and derive their value from future allocation opportunities, such dilution can affect future block outcomes.

To formalize this protection of future purchasing power, we introduce the following counterpart of WBB for a
closed-credit economy. A \emph{nontrader} in block $b$
is an agent $i$ with zero net payment, $T_i^b=0$.

\begin{definition}[Weak Nontrader Neutrality]
A mechanism satisfies \textbf{Weak Nontrader Neutrality (WNN)} if, for every block $b$ and every nontrader $i\in I$, the agent's relative credit balance is weakly preserved:
\begin{equation}\label{eq:wnn}
    \frac{w_i^{b+1}}{\sum_{j\in I}w_j^{b+1}}
    \geq
    \frac{w_i^b}{\sum_{j\in I}w_j^b}.
\end{equation}
\end{definition}

WNN ensures that a nontrader's share of total credits
does not decrease following settlement. The following
lemma relates this condition to WBB.

\begin{lemma}[Equivalence between WBB and WNN]\label{lem:equivalence}
Suppose that, in block $b$, there exists at least one nontrader $i\in I$ with $w_i^b>0$. Then WNN in block $b$ holds if and only if WBB in block $b$ holds.
\end{lemma}

\begin{proof}
For any nontrader $i$ with $T_i^b=0$, its credit balance remains unchanged: $w_i^{b+1}=w_i^b$. Summing the credit transitions across all agents gives $\sum_{j\in I}w_j^{b+1}
    =
    \sum_{j\in I}w_j^b-\sum_{j\in I}T_j^b$.
Since $w_i^b>0$, WNN for this nontrader is equivalent to
\begin{align}
   & \frac{w_i^{b+1}}{\sum_{j\in I}w_j^{b+1}}
    \geq
    \frac{w_i^b}{\sum_{j\in I}w_j^b}\label{eq:WBB_WNN}
   \\ &\iff
    \frac{w_i^b}
    {\sum_{j\in I}w_j^b-\sum_{j\in I}T_j^b}
    \geq
    \frac{w_i^b}{\sum_{j\in I}w_j^b}
    \iff
    \sum_{j\in I}T_j^b\geq0,\nonumber
\end{align}
which is precisely WBB.
\end{proof}



\textbf{Non-myopic Pareto efficiency.} 
We adapt \cite[Definition~2.3]{DOBZINSKI2012486} to
account for agents' predicted value-to-go.
The resulting property rules out feasible changes
that make any agent or the auctioneer strictly
better off without making another participant
worse off.

\begin{definition}[Non-myopic Pareto Efficiency] 
\label{def: NPE}
An allocation $\{x_i^b\}_{i \in I}$ and net payment $\{T_i^b\}_{i \in I}$  in block $b$ is  \textbf{Non-myopic Pareto Efficient (NPE)} if it is feasible and 
there exists no alternative feasible allocation $\{\tilde x_i^{b}\}_{i \in I}$ and net payment $\{\tilde T_i^b\}_{i \in I}$ 
such that
\begin{equation}\label{eq:pareto_efficiency}
    \sum_{k=1}^{\tilde x_i^b}v_i^b(k)
    +
    J_i^{b+1}\!\left(w_i^b-\tilde T_i^b\right)
    \geq
    \sum_{k=1}^{x_i^b}v_i^b(k)
    +
    J_i^{b+1}\!\left(w_i^b-T_i^b\right)
\end{equation}
for every $i\in I$, and
\begin{equation}\label{eq:auctioneer_pareto}
    \sum_{i\in I}\tilde T_i^b
    \geq
    \sum_{i\in I}T_i^b,
\end{equation} 
with at least one of these inequalities being strict.
\end{definition}

\textbf{Non-myopic incentive compatibility.} The following property requires truthful reporting to maximize an agent's
expected utility from its current allocation and
remaining credits, given truthful reports by the
other agents. 
\begin{definition}[Non-myopic Incentive Compatibility]
Let $\mathbf v_i^b= (v_i^b(1),\ldots,v_i^b(N_i^b))$ denote agent $i$'s true marginal valuation profile
in block $b$, and let $\hat{\mathbf v}_i^b$ denote an arbitrary report. Also, let $\mathbf v_{-i}^b$ denote the true valuation profiles by all other agents. 
For any report profile
$(\hat{\mathbf v}_i^b,\mathbf v_{-i}^b)$, let
$x_i^b(\hat{\mathbf v}_i^b,\mathbf v_{-i}^b)$
and
$T_i^b(\hat{\mathbf v}_i^b,\mathbf v_{-i}^b)$
denote, respectively, the allocation and net credit payment assigned to agent $i$
by the mechanism. 
A mechanism is \textbf{Non-myopic Incentive Compatible (NIC)} if for every $i$ and block $b$, truthful reporting maximizes agent $i$'s expected non-myopic utility over the other agents’ valuation profiles $\mathbf v_{-i}^b$:
\begin{equation}\label{eq:block_ic}
\begin{split}
&\mathbb E_{\mathbf v_{-i}^b}\!\left[
    \sum_{k=1}^{x_i^b(\mathbf v_i^b,\mathbf v_{-i}^b)}
        v_i^b(k)
    +
    J_i^{b+1}\!\left(
        w_i^b-T_i^b(\mathbf v_i^b,\mathbf v_{-i}^b)
    \right)\;\Big|\;\mathbf v_i^b
\right]
\\
&~~\geq
\mathbb E_{\mathbf v_{-i}^b}\!\left[
    \sum_{k=1}^{x_i^b(\hat{\mathbf v}_i^b,\mathbf v_{-i}^b)}
        v_i^b(k)
    +
    J_i^{b+1}\!\left(
        w_i^b-T_i^b(\hat{\mathbf v}_i^b,\mathbf v_{-i}^b)
    \right)\;\Big|\;\mathbf v_i^b
\right],
\end{split}
\end{equation}
for every true valuation profile $\mathbf v_i^b$ and every report $\hat{\mathbf v}_i^b$, where expectations are conditional on $\mathbf v_i^b$ and are taken over the other agents' types and the mechanism's randomization.
\end{definition}

\subsection{An Impossibility Result}
Now, we show that there is no mechanism that satisfies all the properties in \cref{sec: normative}. Our result is a closed-credit analogue of the Myerson--Satterthwaite impossibility theorem \cite{myerson1983efficient}.

\begin{theorem}[Impossibility]\label{thm:impossibility}
Consider a block consisting of a buyer, a seller, and a nontrader with a positive credit balance. If the buyer's and seller's valuations are independently distributed with overlapping supports, no mechanism can simultaneously satisfy NIR, WNN, NPE, and NIC.
\end{theorem}

\begin{proof}
It suffices to establish impossibility for the base case where the predicted-value-to-go function $J_i^{b+1}(w)$ is affine in $w$. Under affine predicted-value-to-go function, the relationship~\eqref{eq:quasilinear_marginal} shows that the setting reduces to a standard bilateral trade environment with quasilinear preferences and normalized valuations. Furthermore, by Lemma~\ref{lem:equivalence}, WNN implies WBB in the presence of the positive-balance nontrader. Thus, any mechanism satisfying NPE, NIC, NIR, and WNN would induce a bilateral trading mechanism satisfying Pareto efficiency,  incentive compatibility, individual rationality, and weak budget balance under quasilinear preferences, contradicting the Myerson--Satterthwaite theorem \cite{myerson1983efficient}.
\end{proof}

Thus, a closed-credit economy does not in general
avoid the classical impossibility of jointly achieving
efficiency, incentive compatibility, individual
rationality, and budget balance.

\section{Mechanism Design}
\label{sec: mechanism design}
Motivated by \cref{thm:impossibility}, we prioritize
NIR and WNN: NIR protects participation in the \emph{current
block}, while WNN preserves nontraders' shares of total
credits \emph{across blocks}. We propose an ascending auction with subgrouping that satisfies these requirements
and analyze its efficiency and incentive guarantees.

\subsection{Proposed Mechanism}
Our mechanism has three steps: group splitting procedure, subgroup-based ascending auctions, and aggregate allocation. We present each step below, and \cref{alg:main_auction} summarizes the mechanism.

\subsubsection{Group Splitting Procedure}\label{sec:group_splitting}
For each block $b$, the procedure $\textsc{SplitAgentsAndCapacities}(b,G)$ partitions the agents into $G$ balanced subgroups $\{I_g\}_{g=1}^G$ and assigns integer subgroup capacities that sum to $\bar C_b$. Specifically, the procedure runs as follows.
\begin{enumerate}
    \item Sort agents in descending order of scaled endowments $\bar e_i^b$. When there is a tie, use the credit balance $w_i^b$ to break it. 
    \item Sequentially assign the agents to the subgroup with the lowest aggregate scaled endowment.
    \item Each subgroup $I_g$ is then assigned $\lfloor\widetilde C_{b,g}\rfloor$ items , where $\widetilde C_{b,g}$ is the target capacity: $\widetilde C_{b,g}
    =
    \bar C_b \cdot (\sum_{i\in I_g}\bar e_i^b) / (\sum_{i\in I}\bar e_i^b)$.
    \item The remaining $R = \bar C_b-\sum_{g=1}^G\lfloor\widetilde C_{b,g}\rfloor$ items are assigned according to the Largest Remainder (Hamilton) method. That is, they are assigned one at a time to the subgroups with the largest fractional remainders.
\end{enumerate}
We denote the final capacity by $\bar C_{b,g}$. The above procedure guarantees that $\sum_{g=1}^G\bar C_{b,g}=\bar C_b$.



\subsubsection{Subgroup-Based Ascending Auction}\label{sec:ascending_auction}
After constructing the subgroups, we run an ascending price-discovery procedure separately within each subgroup.

In block $b$, for a given price $p\ge 0$, agent $i$ reports its integer demand $D_i^b(p)$. Any valid demand submission must satisfy domain and budget constraints in \cref{def:feas_alloc}: $D_i^b(p) \in \{0,\ldots,N_i^b\}$ and $(D_i^b(p)-\bar e_i^b)p\leq w_i^b$. When agent $i$ is truthful, its reported demand corresponds to an optimal solution to its dynamic choice problem,
\begin{equation}\label{eq:mechanism_demand}
D_i^b(p)
\in
\argmax_{K\in\{0,\ldots,N_i^b\}}
\left\{
    \sum_{k=1}^K v_i^b(k)
    +
    J_i^{b+1}\!\left(w_i^b-(K-\bar e_i^b)p\right)
\right\}
\end{equation}
subject to $(K-\bar e_i^b)p\leq w_i^b$. If the maximizer is not unique, the agent selects the largest maximizer.  We let $D_g^b(p) = \sum_{i \in I_g} D_i^b(p)$ denote the aggregate demand.

We require each agent's reported demand $D_i^b(p)$ (and thus aggregate demand $D_g^b(p)$) to be nonincreasing in $p$. Under truthful reporting, agents automatically satisfy this requirement whenever $v_i^b(k)$ is nonincreasing in $k$ and $J_i^{b+1}(w)\in \mathcal{J}$, as established in \cref{prop:optimal_demand_price_monotonicity}.
Leveraging this price monotonicity, the procedure $\textsc{SubgroupPriceDiscovery}(I_g,\bar C_{b,g},b,\Delta p)$ runs as follows.
\begin{enumerate}
    \item $p \leftarrow 0$.
    \item Consider three situations and increase the price when there is over demand.
    \begin{enumerate}
        \item If $D_g^b(p) < \bar{C}_{b,g}$, then return $\max\{0, p - \Delta p\}$.
        \item If $D_g^b(p) = \bar{C}_{b,g}$, then return $p$.
        \item If $D_g^b(p) > \bar{C}_{b,g}$, then $p\gets p+\Delta p$ and repeat.
    \end{enumerate}
\end{enumerate}
Let $p_{b,g}$ denote the resulting price in subgroup $I_g$.

\begin{algorithm}[t]
\caption{Block Dynamic Ascending Auction with Subgroup Rationing}
\label{alg:main_auction}
\begin{algorithmic}[1]
\algrenewcommand\algorithmicindent{1.2em}
\Procedure{RunBlockAuction}{$b,G,\Delta p$}
    \State $\{I_g,\bar C_{b,g}\}_{g=1}^G
    \gets\textsc{SplitAgentsAndCapacities}(b,G)$

    \For{$g=1$ \textbf{to} $G$}
        \State $p_{b,g}\gets
        \textsc{SubgroupPriceDiscovery}
        (I_g,\bar C_{b,g},b,\Delta p)$
    \EndFor
    \State $p_b^*\gets\min_{1\leq g\leq G}p_{b,g}$

    \For{$g=1$ \textbf{to} $G$}
        \State $D_g^b(p_b^*)\gets\sum_{i\in I_g}D_i^b(p_b^*)$
        \If{$D_g^b(p_b^*)\leq\bar C_{b,g}$}
            \State $x_i^{b,*}\gets D_i^b(p_b^*),
            \quad\forall i\in I_g$
        \Else
            \For{\textbf{each} $i \in I_g$}
                \State $x_i^{\text{guar}} \gets \min\{D_i^b(p_b^*), \lfloor \bar{e}_i^b \rfloor\}$ 
                \State $x_i^{b,*}\gets x_i^{\mathrm{guar}}$
            \EndFor
            \State $R_g\gets
            \bar C_{b,g}
            -\sum_{i\in I_g}x_i^{\text{guar}}$
            \State Draw a random ordering $\pi_g$ of the agents in $I_g$,
             \Statex \hspace{11.7mm} independently of their reported demands.
            \For{\textbf{each} $i$ in the order $\pi_g$}
            \State $y_i\gets
            \min\{D_i^b(p_b^*)-x_i^{\mathrm{guar}},\,R_g\}$
            \State $x_i^{b,*}\gets x_i^{b,*}+y_i, \quad R_g\gets R_g-y_i$
            \If{$R_g=0$}
                \State \textbf{break}
            \EndIf
        \EndFor
        \EndIf
    \EndFor

    \For{\textbf{each} $i\in I$}
        \State $T_i^b\gets
        (x_i^{b,*}-\bar e_i^b)p_b^*$
        \State $\widehat w_i^{b+1}\gets w_i^b-T_i^b$
    \EndFor

    \For{\textbf{each} $i\in I$}\footnotemark
        \State $w_i^{b+1}\gets
        \widehat w_i^{b+1}
        \dfrac{\sum_{j\in I}w_j^b}
              {\sum_{j\in I}\widehat w_j^{b+1}}$
    \EndFor

    \State \Return
    $\{x_i^{b,*}\}_{i\in I},p_b^*, \{T_i^{b}\}_{i\in I}, \{w_i^{b+1}\}_{i\in I}$
\EndProcedure
\end{algorithmic}
\end{algorithm}
\footnotetext{Since a uniform rescaling of all credit balances preserves relative purchasing power of all agents, the regulatory authority may normalize end-of-block balances to maintain a stationary aggregate credit supply across sequential blocks.}

\subsubsection{Aggregate Allocation}
After the ascending auctions for all the subgroups, we set the common transaction price to the minimum subgroup price, $p_b^*=\min_{1\leq g\leq G}p_{b,g}$. Because $p_b^*$ is no larger than any subgroup discovery price, the demands at $p_b^*$ have already been encountered during the corresponding ascending procedures and may be reused.

At the common unit price $p_b^*$, a subgroup whose aggregate demand does not exceed its assigned capacity receives its demanded allocation. If a subgroup is oversubscribed, each agent first receives up to the integer part of its scaled endowment, subject to its demand. The remaining capacity is rationed by random ordering rule among the agents' excess demands. \cref{alg:main_auction} describes the whole procedure.

\subsubsection{Block Transitions}
After running \cref{alg:main_auction} in block $b$,
we update each agent's maximum resource requirement
for the next block. Agent $i$ has
$N_i^b-x_i^{b,*}$ units of unmet demand, which it
may withdraw or carry forward. Let $\ell_i^b$ denote
the number of units carried forward, where
$\ell_i^b\in\{0,\ldots,N_i^b-x_i^{b,*}\}$. We set $N_i^{b+1}=e_i^{b+1}+\ell_i^b$. Thus, the next block's maximum resource requirement
combines the agent's initial endowment for that
block with its carried-over unmet demand. If all
unmet demand is carried forward, this becomes
$N_i^{b+1}=e_i^{b+1}+N_i^b-x_i^{b,*}$.


\subsection{Theoretical Analysis}

\subsubsection{Feasibility and Market Clearance}
We first establish that the mechanism produces a feasible allocation-price pair (see Definition \ref{def:feas_alloc}) and clears the block capacity, i.e., no capacity is left unallocated whenever $p_b^*>0$. 
\begin{theorem}
\label{thm:mechanism_properties}
If each agent's reported demand $D_i^b(p)$ follows \eqref{eq:mechanism_demand},
\cref{alg:main_auction} gives a feasible allocation-price pair $(\{x_i^{b,*}\}_{i\in I}, p_b^*)$ that achieves market clearance: $p_b^* > 0$ implies $\sum_{i\in I}x_i^{b,*} = \bar C_b$.
\end{theorem}

\subsubsection{Individual Rationality, Nontrader Protection, and Efficiency}
We establish NIR and WNN for every number of subgroups
$G$ and provide sufficient conditions for NPE when $G=1$. Note that, by \cref{thm:impossibility}, NIR and WNN cannot
in general be combined with both NPE and NIC.

\begin{theorem}
\label{thm:mechanism_properties2}
Consider \cref{alg:main_auction} under the conditions
of \cref{thm:mechanism_properties}.
Suppose that the assumptions of
\cref{prop:optimal_demand} and
\cref{lem:equivalence} hold. Then:
\begin{enumerate}
    \item \textbf{NIR and WNN.}
    The mechanism satisfies NIR and WNN for every
    number of subgroups $G$.

    \item \textbf{NPE under optimal-demand clearing.}
    \label{item:unified_npe}
    If $G=1$, each predicted-value-to-go
    $J_i^{b+1}$ is strictly increasing,  all agents
report their demands truthfully, and every agent receives their reported demand\footnote{In this case, NPE holds if either $\sum_{i\in I}D_i^b(0)\leq\bar C_b$ or $\sum_{i\in I}D_i^b(k\Delta p)=\bar C_b$ for some $k\in\mathbb Z_{\geq0}$, since all reported demands are fully allocated. Otherwise, price-grid jumps may require rationing, which need not preserve individual optimality or NPE. Designing efficiency-preserving rationing rules is left for future work.} at price $p_b^*$, then the allocation and associated transfers
    $T_i^b=(x_i^{b,*}-\bar e_i^b)p_b^*$
    are NPE.
\end{enumerate}
\end{theorem}

\subsubsection{Approximate Incentive Compatibility}
We quantify the incentive effects of subgrouping by
bounding the expected gain from misreporting in terms
of the probability that an agent changes the common price.

Fix a block $b$, its public state, and the
predicted-value-to-go functions. 
For a report profile $\hat{\mathbf v}^b$, define
agent $i$'s predicted non-myopic utility, evaluated
at its true valuations, as
\[
U_i^b(\hat{\mathbf v}^b;\mathbf v_i^b)
=
\sum_{k=1}^{x_i^{b,*}(\hat{\mathbf v}^b)}v_i^b(k)
+
J_i^{b+1}\!\left(
w_i^b-T_i^b(\hat{\mathbf v}^b)
\right).
\]
Expectations are taken over the other agents' types
and the mechanism's random ordering.

\begin{definition}[Approximate NIC]
A mechanism is $\epsilon$-NIC in block $b$ if,
for every agent $i$, true type $\mathbf v_i^b$,
and admissible report $\hat{\mathbf v}_i^b$,
\[
\begin{aligned}
&\mathbb E\!\left[
U_i^b((\hat{\mathbf v}_i^b,\mathbf v_{-i}^b);
       \mathbf v_i^b)
-
U_i^b((\mathbf v_i^b,\mathbf v_{-i}^b);
       \mathbf v_i^b)
\,\big|\,\mathbf v_i^b
\right]
\leq\epsilon.
\end{aligned}
\]
Thus, truthful reporting is optimal up to 
expected utility gain $\epsilon$.
\end{definition}

Define the maximum probability that a unilateral
misreport changes the common price by
\[
q_G^b
:=
\sup_{i,\mathbf v_i^b,\hat{\mathbf v}_i^b}
\mathbb P\big(
p_b^*(\hat{\mathbf v}_i^b,\mathbf v_{-i}^b)
\neq
p_b^*(\mathbf v_i^b,\mathbf v_{-i}^b)
\,\big|\,\mathbf v_i^b
\big),
\]
where the supremum is over admissible types and reports.

\begin{theorem}[Approximate NIC under subgrouping]
\label{thm:approximate-NIC}
Consider \cref{alg:main_auction} under the conditions of
\cref{thm:mechanism_properties,prop:optimal_demand}.
Suppose that individual total valuations satisfy the uniform bound $\sup_{i\in I,\mathbf v_i^b}
\sum_{k=1}^{N_i^b}v_i^b(k)
\leq \overline V_b<\infty$, and each predicted-value-to-go $J_i^{b+1}$ is
finite-valued and nondecreasing. Define
\[
W_b:=\sum_{i\in I}w_i^b,
\qquad
M_b:=\overline V_b+
\max_{i\in I}
\left\{J_i^{b+1}(W_b)-J_i^{b+1}(0)\right\}.
\]
Then the mechanism is $(M_bq_G^b)$-NIC.
For a fixed market, $M_b$ is independent of $G$.
In particular, if $q_G^b\leq c/G$ for a constant $c$
independent of $G$, the mechanism is $(M_bc/G)$-NIC.
\end{theorem}

The bound $M_bq_G^b$ explains how increasing $G$ can
improve incentive compatibility: additional subgroups
provide competing prices beyond an individual agent's
control, potentially reducing its influence on the
common price. Since $p_b^*=\min_g p_{b,g}$, a deviation
can change this price only if the agent's subgroup
attains the truthful minimum or its deviation lowers
its subgroup price below that minimum. When
$q_G^b\leq c/G$, the resulting guarantee improves
with $G$ for a fixed market.

This improvement is conditional, since smaller
subgroups may also be more sensitive to individual
reports. Balancing endowments alone therefore does
not guarantee that $q_G^b$ decreases with $G$.

\section{Experiments}
We simulate the sequential block auction market for arrival slots at San Francisco International Airport (SFO) on August 20, 2026.
The arrival schedule published the day before \cite{flightstats_sfo} (591 arrivals, 52 airlines) serves as the endowment and as a proxy for the months-ahead allocation under current practice. Nominal capacity is 52 arrivals/hr \cite{faa_sfo_capacity_2019}; we consider a mild (S1) and a severe (S2) capacity reduction (Fig.~\ref{fig:allocations-per-hour}).

\paragraph{Delay and cancellation costs.} Aircraft type and seats $S_f$ come from BTS T-100 Segment data, matched by operating carrier and origin--SFO route \cite{bts_transtats}. For arrival delay $d$ (assigned slot minus scheduled time, in minutes), the cost of flight $f$ is
\[
c_f(d) = \rho_f\,(d - g)^+ + \rho_f\,(d - b_f)^+ .
\]
The first term is the cost of this arrival and the second is the delay propagated to the aircraft's next departure. $\rho_f$ is the airline's operating cost per minute for the aircraft type (BTS Form 41, Schedules P-5.2 and P-7 \cite{bts_transtats}) and $g = 15$ min is the on-time threshold \cite{bts_ontime}. We assume a buffer time $b_f \sim U[15,20]$ min, the slack the aircraft can absorb before the delay cuts into its turnaround and propagates to the next departure. We assume a flight is canceled once its delay exceeds 3 hours, the U.S. DOT threshold for a significant delay on domestic itineraries, beyond which passengers are entitled to a refund~\cite{dot_refunds}. The cancellation cost is a refund of \$305 per passenger, the average U.S. domestic one-way itinerary fare~\cite{bts_fares_2026q1}.

\begin{figure}
    \centering
    \includegraphics[width=1\linewidth]{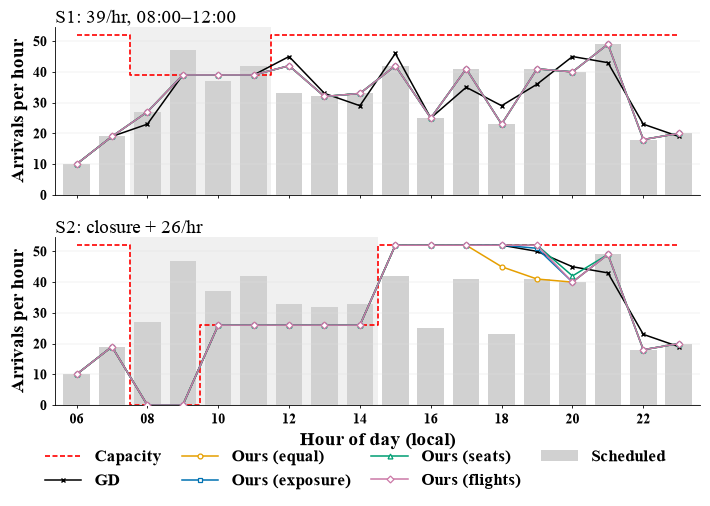}
    \caption{Scheduled arrivals per hour at SFO (bars) against capacity (dashed red). Shading marks the capacity reduction: in S1 (top), capacity drops to 39 arrivals/hr from 08:00 to 12:00; in S2 (bottom), zero capacity from 08:00 to 10:00 and reduced capacity to 26 arrivals/hr until 15:00. Lines show realized arrivals per hour under the ground delay baseline (GD) and our mechanism with each budget initialization.}
    \label{fig:allocations-per-hour}
\end{figure}

\paragraph{Slot valuation and credits.} Airline $i$'s value for its $k$th slot in block $b$ (Sec.~\ref{sec:value}) is the cost it avoids by not delaying to $b+1$, $v_i^b(k) = c_f(d^{b+1}) - c_f(d^b)$. $f$ is the airline's flight with the $k$th highest avoided cost, and $d^b$ is its delay when landing at the midpoint of block $b$. Initial credits are assigned to each airline $i$ in proportion to weight $\omega_i$, $W = |I|\bar w$ as $w_i^1 = W\omega_i/\sum_{j}\omega_j$. We compare four weights $\omega_i$:
\emph{Equal} ($\omega_i = 1$);
\emph{Exposure}, the cost of delaying all of $i$'s flights by one block, $\omega_i = U_i = \sum_{f\in F_i} c_f(B)$;
\emph{Seats}, $i$'s total number of seats (or passengers it carries), $\omega_i = \sum_{f\in F_i} S_f$; and
\emph{Flights}, $i$'s total number of flights, $\omega_i = |F_i|$.



\paragraph{\textbf{Results.}} Table~\ref{tab:budget-rules} compares our mechanism under each credit initialization with the ground delay baseline (GD), which assigns slots by ration-by-schedule followed by compression \cite{vossen_ball_2006}. In S1, no flights are canceled, and our mechanism improves on GD in every metric under every initialization. \textit{Exposure} yields the lowest total cost (64\% below GD), while \textit{seats} yields the least delay beyond 15 min and fewest propagated delays. In S2, our mechanism lowers total cost by 7--18\% under every initialization but cancels at least as many flights as GD. \textit{Exposure, seats, and flights} cancel about as many flights as GD (31--34 vs.\ 30) while lowering delay cost by 7--10\%; \textit{exposure} and \textit{seats} also cancel fewer passengers (3.3k and 2.7k vs.\ 4.0k), i.e., smaller flights. \textit{Equal} attains the lowest total cost by canceling the most flights (51), which reduces propagated delays and delays beyond 15 min by $\sim$18\%. Whether this trade is desirable depends on how airlines value cancellations. We value a cancellation only by the refund owed to its passengers; in practice, airlines also weigh rebooking, crew and aircraft repositioning, and reputation, which would raise this cost. 

These magnitudes depend on parameters we estimate from public data: operating cost per minute by aircraft type, turnaround slack, and cancellation costs. Our goal is not to predict exact savings but to show that a market informed by heterogeneous costs, in which airlines trade, improves on current practice. GD assigns slots in schedule order and ignores costs. In our mechanism, each airline bids using its own costs and trading allows slots to move to flights with higher delay costs.

\begin{table}[t]
\centering
\caption{Day-of auction outcomes by budget initialization.}
\label{tab:budget-rules}
\scriptsize
\setlength{\tabcolsep}{2pt}
\setlength{\aboverulesep}{0pt}
\setlength{\belowrulesep}{0pt}
\renewcommand{\arraystretch}{1.1}
\begin{threeparttable}
\begin{tabular*}{\columnwidth}{@{\extracolsep{\fill}} l rrrrr | rrrrr}
\toprule
 & \multicolumn{5}{c|}{S1} & \multicolumn{5}{c}{S2} \\
\cmidrule(lr){2-6} \cmidrule(lr){7-11}
Metric & GD & Eq. & Exp. & Seat & Flt & GD & Eq. & Exp. & Seat & Flt \\
\midrule
Delay cost ($\$$k) & 72.2 & 42.5 & 25.9 & 29.8 & 27.1 & 4,855 & 2,987 & 4,386 & 4,532 & 4,401 \\
Cancel. cost ($\$$k) & 0 & 0 & 0 & 0 & 0 & 1,222 & 2,012 & 1,018 & 820 & 1,252 \\
Total cost ($\$$k) & 72.2 & 42.5 & 25.9 & 29.8 & 27.1 & 6,077 & 4,999 & 5,404 & 5,353 & 5,653 \\
Canceled & 0 & 0 & 0 & 0 & 0 & 30 & 51 & 34 & 31 & 33 \\
Delay $>$15 (min) \tnote{a} & 78 & 33 & 28 & 23 & 43 & 352 & 289 & 344 & 347 & 346 \\
Pax delayed (k)\tnote{b} & 11.0 & 4.8 & 3.2 & 2.7 & 6.1 & 47.9 & 38.5 & 48.0 & 48.6 & 47.4 \\
Pax canceled (k)\tnote{c} & 0 & 0 & 0 & 0 & 0 & 4.0 & 6.6 & 3.3 & 2.7 & 4.1 \\
Propagated\tnote{d} & 66 & 21 & 19 & 17 & 28 & 352 & 286 & 343 & 347 & 346 \\
\bottomrule
\end{tabular*}
\begin{tablenotes}[flushleft]
\scriptsize
\item[] GD: ground delay baseline (RBS + compression). Ours with budget initialization methods: equal, exposure, seats, and flights (Eq., Exp., Seat, Flt). Cost (USD) and number of passengers are in thousands.
\item[a] Number of delays beyond the 15-min grace period.
\item[b] Passengers on landed flights delayed $>$15 min; canceled passengers excluded.
\item[c] Passengers on canceled flights.
\item[d] Flights whose arrival delay exceeds their turnaround buffer ($d_f > b_f$), so the delay cascades to the aircraft's next flight.
\end{tablenotes}
\end{threeparttable}
\end{table}



Figure \ref{fig:wnn} depicts the budget share (ratios in \eqref{eq:wnn}) of the airlines in S2. The start of a dark line segment marks a block when the airline did not trade ($T_i^b=0$). These dark segments are always non-decreasing, which shows that non-participation does not \textit{decrease} their budget share, validating WNN. We also report that the smallest value gained from participating in the market minus a rebate value for each partial unit at the clearing price each block is 0, which is non-negative thus validating NIR.

\begin{figure}
    \centering
    \includegraphics[width=\linewidth]{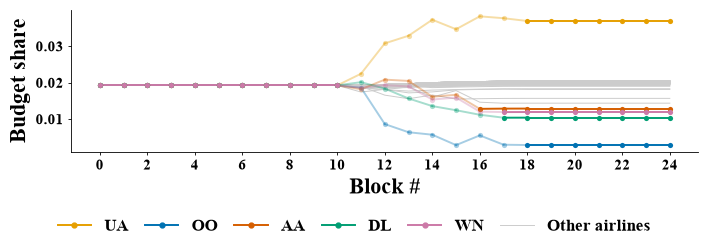}
    \caption{Relative credit balance $w_i^b/\sum_j w_j^b$ by block (S2, equal budgets), for the five airlines with the largest credit flows; gray lines show the remaining airlines. Dark segments mark blocks in which the airline does not trade ($T_i^b = 0$), faded segments mark blocks in which it trades.}
    \label{fig:wnn}
\end{figure}

Lastly, we show that our allocations satisfy the conditions of a feasible allocation--price pair (Def.~\ref{def:feas_alloc}). In Fig.~\ref{fig:feasibility}(a), clearing prices are nonnegative, rise during the capacity reduction as pending flights accumulate, and peak before the reduction ends. In Fig.~\ref{fig:feasibility}(b), each airline's allocation remains within its pending flights: the ratio never exceeds one, so the domain constraint holds. During the reduction, airlines land 3--45\% of their pending flights per block, and the fraction varies across airlines according to their bids rather than schedule order. Once capacity is restored, the ratio approaches one for most airlines by blocks 16--17, as the backlog clears.
Allocations are always at or below the capacity line in Fig. \ref{fig:allocations-per-hour}, showing the capacity constraints have been respected. Finally, in Fig. \ref{fig:wnn}, we see strictly positive budget shares, which shows an airline's payment never exceeds their budget.

\begin{figure}
    \centering
    \includegraphics[width=\linewidth]{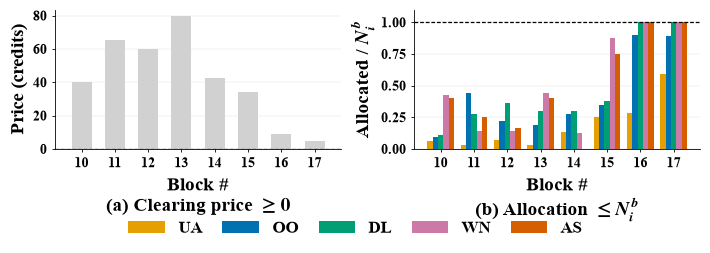}
    \caption{Feasible Allocation-Price Pair (S2, equal budgets, contested blocks only). (a) Clearing price per block, in credits. (b) Slots allocated to each airline as a fraction of its pending flights $N_i^b$, for the five airlines with the most pending flights across contested blocks. Dashed line marks the bound $N_i^b$.}
    \label{fig:feasibility}
\end{figure}



\section{Conclusion}
We study credit-based markets for sequential resource
reallocation under capacity uncertainty and existing
endowments. We establish an impossibility result for
jointly achieving NIR, WNN, NPE, and NIC, and propose
a group-splitting ascending-auction mechanism that preserves NIR and WNN while providing
conditional efficiency and approximate incentive
guarantees. Simulations illustrate its application to
airport slot allocation. Future work will develop
principled subgroup-selection and budget allocation rules and derive sharper efficiency and incentive guarantees. 

\printbibliography

\appendix
\subsection{Proof of \cref{prop:optimal_demand}}

\begin{proof}
 Since $J_i^{b+1}$ is concave, the credit opportunity cost  (the bracketed difference in \eqref{eq:effective_marginal}) is nondecreasing in $k$. Along with $v_i^b(k)$ being nonincreasing, $\Delta_i^b(k,p,w,e)$ is nonincreasing in $k$, so \eqref{eq:effective_positive} characterizes an optimal demand.
\end{proof}

\subsection{Proof of \cref{prop:nonincreasing_in_p}}
\begin{proof}
\emph{Necessity.}
Suppose that every $F_{w,a}$ is nondecreasing on
its feasible price interval.

We first establish that $J$ is nondecreasing on $[0, \infty)$ and continuous on $(0,\infty)$.
Taking $w=0$ and $a=-1$ gives $F_{0,-1}(p)=J(p)-J(0)$, so $J$ is nondecreasing on $[0,\infty)$.

Fix $c>0$. Since a monotone function has at most
countably many discontinuities, we can choose a
continuity point $d\in(0,c)$.
Set
\[
w=2c-d,\qquad a=1,\qquad p_0=c-d.
\]
Then $p_0$ is an interior point of the feasible
interval $[0,w/2]$, and
\[
F_{w,1}(p)=J(w-p)-J(w-2p).
\]
Writing $J(c-)$ and $J(c+)$ for the one-sided
limits of $J$ at $c$, continuity at $d$ gives
\[
\lim_{p\uparrow p_0}F_{w,1}(p)
=
J(c+)-J(d),
\]
and
\[
\lim_{p\downarrow p_0}F_{w,1}(p)
=
J(c-)-J(d).
\]
Monotonicity of $F_{w,1}$ therefore implies
$J(c+)\leq J(c-)$.
Together with $J(c-)\leq J(c)\leq J(c+)$, this proves continuity at $c$. Thus $J$ is continuous on $(0,\infty)$.

We next establish concavity.
Fix $0<u<v$ and $h>0$.
For sufficiently large positive integers $n$, set
\[
a_n=n,\qquad
w_n=u+(n+1)h,\qquad
p_n=h-\frac{v-u}{n+1}.
\]
Both $p_n$ and $h$ are feasible, with
$0\leq p_n\leq h$. Since
\[
w_n-(n+1)p_n=v,
\qquad
w_n-np_n=v+p_n,
\]
monotonicity yields $F_{w_n, n}(p_n) \leq F_{w_n, n}(h)$, so
\[
J(v+p_n)-J(v)
\leq
J(u+h)-J(u).
\]
Letting $n\to\infty$ and using the interior
continuity established above gives
\[
J(v+h)-J(v)
\leq
J(u+h)-J(u).
\]
Taking $v=u+h$ proves midpoint concavity.
Together with interior continuity, this establishes
concavity on $(0,\infty)$.

Because $J$ is nondecreasing and finite-valued,
the limit $L:=\lim_{z\downarrow0}J(z)$ exists and is finite, with $J(0)\leq L$.
For $y>0$ and $t\in(0,1)$, interior concavity gives
\[
J((1-t)\varepsilon+ty)
\geq
(1-t)J(\varepsilon)+tJ(y).
\]
Letting $\varepsilon\downarrow0$ yields
\[
J(ty)
\geq
(1-t)L+tJ(y)
\geq
(1-t)J(0)+tJ(y).
\]
Hence $J$ is concave on $[0,\infty)$.

The right derivative $J'_+(z)$ therefore exists
and is finite for every $z>0$.
Fix $0<x<y$, and set
\[
h=y-x,\qquad r=\frac{x}{h}.
\]
Taking $w=0$ and $a=-(r+1)$ gives the
nondecreasing function
\[
F_{0,-(r+1)}(p)
=
J((r+1)p)-J(rp).
\]
Its right derivative at $p=h$ is nonnegative:
\[
0
\leq
(r+1)J'_+(y)-rJ'_+(x)
=
\frac{yJ'_+(y)-xJ'_+(x)}{y-x}.
\]
Thus $zJ'_+(z)$ is nondecreasing.

\noindent
\emph{Sufficiency.}
Suppose that $J$ is concave on $[0,\infty)$
and $zJ'_+(z)$ is nondecreasing on $(0, \infty)$.

These conditions imply that $J$ is differentiable
on $(0,\infty)$.
Indeed, for every $c>0$, concavity gives
$J'_-(c)\geq J'_+(c)$, whereas monotonicity gives
\[
cJ'_-(c)
=
\lim_{z\uparrow c}zJ'_+(z)
\leq
cJ'_+(c).
\]
Thus the one-sided derivatives coincide.

Moreover, $J'(x)\geq0$ for every $x>0$.
Otherwise, choosing $y>x$ and using concavity
would give $yJ'(y)\leq yJ'(x)<xJ'(x)$, contradicting monotonicity of $zJ'(z)$.
Hence $J$ is nondecreasing on $(0,\infty)$.
Concavity at zero implies $J(0)\leq L:=\lim_{z\downarrow0}J(z)$, and this limit is finite.
Consequently, $J$ is nondecreasing on
$[0,\infty)$ as well.

Fix $w\geq0$ and $a\in\mathbb R$.
At any feasible price $p>0$ for which
\[
x=w-(a+1)p>0,
\qquad
y=w-ap=x+p,
\]
differentiation gives $F'_{w,a}(p)
=
J'(x)+a\bigl(J'(x)-J'(y)\bigr)$. Since $w=y+a(y-x)\geq0$, we have
$a\geq-y/(y-x)$.
Concavity and monotonicity of $zJ'(z)$ yield
\[
\begin{aligned}
F'_{w,a}(p)
&\geq
J'(x)-\frac{y}{y-x}
       \bigl(J'(x)-J'(y)\bigr)\\
&=
\frac{yJ'(y)-xJ'(x)}{y-x}
\geq0.
\end{aligned}
\]

The derivative argument establishes monotonicity wherever
$p>0$ and $w-(a+1)p>0$. We now check the remaining
feasible prices in four cases.
\begin{enumerate}
    \item \textbf{$w>0$.}
    Continuity of $J$ on $(0,\infty)$ gives $\lim_{p\downarrow0}F_{w,a}(p)=0=F_{w,a}(0)$, extending monotonicity to $p=0$.
    If $a+1>0$, the feasible interval also has the
    finite upper endpoint $p_{\max}=w/(a+1)$, where
    \[
    \begin{aligned}
    F_{w,a}(p_{\max})
    &=J(p_{\max})-J(0)\\
    &\geq J(p_{\max})-L\\
    &=\lim_{p\uparrow p_{\max}}F_{w,a}(p).
    \end{aligned}
    \]
    Thus monotonicity is preserved at this endpoint.
    If $a+1\leq0$, there is no finite upper endpoint.

    \item \textbf{$w=0$ and $a>-1$.}
    The feasible interval is the singleton $\{0\}$,
    so monotonicity holds trivially.

    \item \textbf{$w=0$ and $a=-1$.}
    We have $F_{0,-1}(p)=J(p)-J(0)$, which is nondecreasing because $J$ is nondecreasing.

    \item \textbf{$w=0$ and $a<-1$.}
    The derivative argument applies for every $p>0$.
    Moreover, $\lim_{p\downarrow0}F_{0,a}(p)
    =L-L=0=F_{0,a}(0)$, extending monotonicity to $p=0$.
\end{enumerate}
Thus $F_{w,a}$ is nondecreasing on its entire feasible
interval.
\end{proof}

\subsection{Proof of \cref{prop:optimal_demand_price_monotonicity}}
\begin{proof}
Let $D_i^b(p)$ be the largest optimal demand defined in Proposition \ref{prop:optimal_demand}. 
 Fix $p_1\le p_2$.
If $D_i^b(p_2)=0$, the result is trivial. Otherwise,
$\Delta_i^b(D_i^b(p_2),p_2,w,e)\ge0$ from \eqref{eq:effective_positive}.
Since $J_i^{b+1} \in \mathcal{J}$, Proposition \ref{prop:nonincreasing_in_p} implies that $\Delta_i^b(k, p, w, e)$ is nonincreasing in $p$, leading to
$\Delta_i^b(D_i^b(p_2),p_1,w,e)
\ge
\Delta_i^b(D_i^b(p_2),p_2,w,e)
\ge0$.

Moreover, we verify budget feasibility. Since $D_i^b(p_2)$ is feasible at $p_2$, we have $w - (D_i^b(p_2) - e)p_2 \ge 0$. If $D_i^b(p_2) \ge e$, then $p_1 \le p_2$ implies $w - (D_i^b(p_2) - e)p_1 \ge w - (D_i^b(p_2) - e)p_2 \ge 0$. If $D_i^b(p_2) < e$, the agent is a net seller, which ensures feasibility for any $p_1 \ge 0$. This ensures that $D_i^b(p_2)$ is budget-feasible at price $p_1$.

Thus, $D_i^b(p_2)$ fulfills \eqref{eq:effective_positive} under budget feasibility at $p_1$. 
By definition in Proposition \ref{prop:optimal_demand},  $D_i^b(p_1)$ is the largest demand at $p_1$ satisfying \eqref{eq:effective_positive} under budget feasibility; it follows that $D_i^b(p_1) \ge D_i^b(p_2)$.
\end{proof}

\subsection{Proof of \cref{thm:mechanism_properties}}
\begin{proof}
The rule dictates agent demand $D_i^b(p)$ is nonincreasing in $p$, and it satisfies domain and budget constraints. We check the following four conditions.
\paragraph{Price Nonnegativity} By construction, $p_b^*\ge0$.

\paragraph{Domain/Budget Constraints} \cref{alg:main_auction} ensures that $x_i^{b,*} \le D_i^b(p_b^*)$, so $x_i^{b,*} \in \{0, ..., N_i^b\}$ and $
(x_i^{b,*} - \bar e_i^b) p_b^* \le (D_i^b(p_b^*) - \bar e_i^b) p_b^* \le w_i^b$, which holds for all $p_b^* \ge 0$, establishing the budget constraint.

\paragraph{Capacity Constraints}
For each subgroup $g$, if $D_g^b(p_b^*)\le\bar C_{b,g}$, the algorithm assigns
$x_i^{b,*}=D_i^b(p_b^*)$, while if $D_g^b(p_b^*)>\bar C_{b,g}$, it
allocates exactly $\bar C_{b,g}$ items through the guaranteed
allocation and random ordering rule. Since
$\sum_{i\in I_g}x_i^{b,*}\le\bar C_{b,g}$ for all $g\in\{1,\dots,G\}$,
it holds that
$\sum_{i\in I}x_i^{b,*}\le
\sum_{g=1}^G\bar C_{b,g}=\bar C_b.$

\paragraph{Market Clearance} Suppose that $p_b^*>0$. It suffices to prove that $\sum_{i\in I}x_i^{b,*}=\bar C_b$.
 Since
$p_b^*=\min_g p_{b,g}\le p_{b,g}$ and agent demand $D_i^b(p)$ is nonincreasing in $p$,
 the \textsc{SubgroupPriceDiscovery} rule yields $D_g^b(p_b^*)\ge D_g^b(p_{b,g})\ge\bar C_{b,g}$ when $p_b^*>0$. 
Thus, for every subgroup, either $D_g^b(p_b^*)=\bar C_{b,g}$ and its
demand is allocated in full, or $D_g^b(p_b^*)>\bar C_{b,g}$ and the
random ordering rule allocates exactly $\bar C_{b,g}$ items. Thus, 
$\sum_{i\in I_g}x_i^{b,*}=\bar C_{b,g}
$ for all $g\in\{1,\dots, G\}$,
which leads to
$\sum_{i\in I}x_i^{b,*}
=
\sum_{g=1}^G\bar C_{b,g}
=
\bar C_b.$
\end{proof}

\subsection{Proof of \cref{thm:mechanism_properties2}}
\begin{proof}
We prove NIR, WNN, and NPE below.
\paragraph{Proof of NIR} Fix an agent $i$ and observe that 
$T_i^b=(x_i^{b,*}-\bar e_i^b)p_b^*$ and 
$\lfloor\bar e_i^b\rfloor-\bar e_i^b=-r_i^b$.
It follows that
\begin{align}\label{eq:diff}
    \sum_{k=1}^{x_i^{b,*}}\Delta_i^b(k,p_b^*,w_i^b,\bar e_i^b)-\sum_{k=1}^{\lfloor\bar e_i^b\rfloor}\Delta_i^b(k,p_b^*,w_i^b,\bar e_i^b)
\end{align}
 recovers the difference between the left-hand  and the right-hand sides of \eqref{eq:block_ir} (see \eqref{eq:effective_marginal} for the definition of $\Delta_i^b(k,p,w,e)$). 

By Proposition~\ref{prop:optimal_demand}, $\Delta_i^b(k,p,w,e)$ is nonincreasing in $k$. Also, since $D_i^b(p_b^*)$ is budget-feasible and optimal, 
$\Delta_i^b(k,p_b^*,w_i^b,\bar e_i^b)\ge 0$ for all $k\in\{1,\dots, D_i^b(p_b^*)\}$.
To show NIR, we need to show that the term \eqref{eq:diff} is nonnegative for the following two cases.
\begin{enumerate}
    \item $D_i^b(p_b^*)\ge \lfloor\bar e_i^b\rfloor$. By the rationing rule, $\lfloor\bar e_i^b\rfloor\le x_i^{b,*}\le D_i^b(p_b^*)$. Hence, 
we have 
$\Delta_i^b(k,p_b^*,w_i^b,\bar e_i^b)\ge0$ for all $k\in\{\lfloor\bar e_i^b\rfloor+1,\dots,x_i^{b,*}\}$. Thus, \eqref{eq:diff} is nonnegative.

    \item $D_i^b(p_b^*)<\lfloor\bar e_i^b\rfloor$. In this case, $x_i^{b,*}=D_i^b(p_b^*)$. Since $D_i^b(p_b^*)$ is an optimal demand and $\lfloor\bar e_i^b\rfloor$ is budget-feasible, we get $\Delta_i^b(k,p_b^*,w_i^b,\bar e_i^b)<0$ for all
$k\in\{D_i^b(p_b^*)+1,\dots,\lfloor\bar e_i^b\rfloor\}$ by Proposition \ref{prop:optimal_demand}.
This case also yields \eqref{eq:diff} to be nonnegative.
\end{enumerate}

\paragraph{Proof of WNN} By \cref{lem:equivalence}, WNN and WBB are equivalent. Thus, it suffices to prove that $\sum_{i\in I}(x_i^{b,*}-\bar e_i^b)p_b^* \geq0$ (see \eqref{eq:wbb}).
If $p_b^*=0$, the inequality is trivial. Suppose that $p_b^*>0$. Then, Theorem \ref{thm:mechanism_properties} implies that $ \sum_{i\in I}x_i^{b,*}=\bar C_b$. Since $\sum_{i\in I} \bar e_i^b\le \bar C_b$, we have $\sum_{i\in I}(x_i^{b,*}-\bar e_i^b)\ge 0$. 

\paragraph{Proof of NPE when $G=1$}
By \cref{thm:mechanism_properties}, the resulting allocation-price pair
is feasible and satisfies $p_b^*\left(\bar C_b-\sum_{i\in I}x_i^{b,*}\right)=0$. Suppose, toward a contradiction, that a feasible alternative
$(\{\tilde x_i^b\}_{i\in I},\{\tilde T_i^b\}_{i\in I})$
Pareto dominates the outcome in the sense of \cref{def: NPE}.

Optimality of $x_i^{b,*}$ at price $p_b^*$ and strict
monotonicity of $J_i^{b+1}$ imply that $\tilde T_i^b
\leq p_b^*(\tilde x_i^b-\bar e_i^b)$ for every $i \in I$. If not, purchasing
$\tilde x_i^b$ at price $p_b^*$ would be budget-feasible
and leave strictly more credits than the alternative.
It would therefore yield strictly higher utility than
the alternative, contradicting optimality of $x_i^{b,*}$.
Moreover, the displayed inequality must be strict
whenever agent $i$ strictly prefers the alternative.

Summing these inequalities and using feasibility and
market clearing yields
\[
\begin{aligned}
\sum_{i\in I}\tilde T_i^b
&\leq
p_b^*\left(\sum_{i\in I}\tilde x_i^b
          -\sum_{i\in I}\bar e_i^b\right)
\leq
p_b^*\left(\bar C_b-\sum_{i\in I}\bar e_i^b\right)
=
\sum_{i\in I}T_i^b.
\end{aligned}
\]
Thus, the auctioneer cannot be strictly better off.
If any agent is strictly better off, the first
inequality is strict, making the auctioneer strictly
worse off. Both cases contradict Pareto
dominance.
\end{proof}

\subsection{Proof of \cref{thm:approximate-NIC}}
\begin{proof}
For each $G$, the group splitting procedure in \cref{sec:group_splitting} ensures that the subgroup partition and capacities are independent of agents' reports.
Fix an agent $i$, its true type $\mathbf v_i^b$,
and a deviation $\hat{\mathbf v}_i^b$.
Compare truthful reporting and the deviation using
the same realization of the other agents' types
and the same random ordering.
Let $E_i$ denote the event that the deviation changes
the common price.

On $E_i^c$, the common price is unchanged, say $p$.
The other agents' demands at $p$ are therefore unchanged.
For fixed opponents' demands and ordering, the
allocation rule can be written as $x_i(d)=\min\{d,K_i\}$, where $d$ is agent $i$'s reported demand at $p$
and $K_i$ does not depend on its report.
To see this, let $r_j=\min\{D_j^b(p),\lfloor\bar e_j^b\rfloor\}$ for $j \neq i$. For $i\in I_g$, the cap is
\[
K_i
=
\max\left\{
\lfloor\bar e_i^b\rfloor,\,
\bar C_{b,g}
-\sum_{j\in I_g\setminus\{i\}}r_j
-\sum_{\substack{j\in I_g\\j\prec i}}
       \bigl(D_j^b(p)-r_j\bigr)
\right\},
\]
where $j\prec i$ means that $j$ precedes $i$
in the random ordering.

By \cref{prop:optimal_demand}, agent $i$'s true
utility at price $p$ has nonincreasing marginal gains.
Consequently,
$\min\{D_i^b(p),K_i\}$ maximizes its utility
over the budget-feasible quantities not exceeding
$K_i$. Truthful reporting therefore yields weakly
higher utility than any deviation on $E_i^c$.

To bound the gain on $E_i$, observe that budget
feasibility and weak budget balance imply, under
either report,
\[
0\leq w_i^b-T_i^b
\leq \sum_{j\in I}(w_j^b-T_j^b)
= W_b-\sum_{j\in I}T_j^b
\leq W_b.
\]
Since valuations are nonnegative and
$J_i^{b+1}$ is nondecreasing, every resulting utility
lies in $\left[
J_i^{b+1}(0),\,
\overline V_b+J_i^{b+1}(W_b)
\right]$. Thus, the gain from deviation is at most $M_b$
on $E_i$ and is nonpositive on $E_i^c$, so
\[
\begin{aligned}
&\mathbb E\!\left[
U_i^b((\hat{\mathbf v}_i^b,\mathbf v_{-i}^b);
       \mathbf v_i^b)
-
U_i^b((\mathbf v_i^b,\mathbf v_{-i}^b);
       \mathbf v_i^b)
\,\middle|\,\mathbf v_i^b
\right]\\
&\qquad\leq
M_b\,\mathbb P(E_i\mid\mathbf v_i^b)
\leq M_bq_G^b.
\end{aligned}
\]
Taking the supremum over agents, true types, and
admissible deviations proves the claim.
\end{proof}

\end{document}